\documentclass[10pt,draftcls,onecolumn]{IEEEtran}

\usepackage[utf8]{inputenc}
\usepackage[T1]{fontenc}
\usepackage{url}
\usepackage{ifthen}
\usepackage{cite}
\usepackage[cmex10]{amsmath} 

\usepackage{graphics} 
\usepackage[pdftex]{graphicx}
\usepackage{epsfig} 
\usepackage{mathptmx} 
\usepackage{times} 
\usepackage{amsmath} 
\usepackage{amssymb}  
\usepackage{subfigure}
\usepackage{verbatim}
\usepackage{cite}
\usepackage{epstopdf}
\usepackage{color}
\usepackage{relsize}
\usepackage{nopageno}
\usepackage{color}
\usepackage{amsthm}

\newtheorem{theorem}{Theorem}

\newtheorem{assumption}{Assumption}
\newtheorem{remark}{Remark}
\newtheorem{lemma}{Lemma}
\newtheorem{corollary}{Corollary}
\newtheorem{proposition}{Proposition}
\newtheorem{definition}{Definition}

\title{Secrecy Capacity of Colored Gaussian Noise Channels with Feedback}

\author{\quad Chong Li, Yingbin Liang, H. Vincent Poor and Shlomo Shamai (Shitz) %
\thanks{Chong Li is with Qualcomm Research, Bridgewater, NJ, 08807, USA (chongl@qti.qualcomm.com). Yingbin Liang is with the Department of Electrical and Computer Engineering, The Ohio State University, Columbus, OH, 43220, USA (email: liang.889@osu.edu). H. Vincent Poor is with the Department of Electrical Engineering, Princeton University, Princeton, NJ, 08544, USA (email: poor@princeton.edu). Shlomo Shamai (Shitz) is with the Department of Electrical Engineering, Technion-Israel Institute of Technology, Technion City, Haifa 32000, Israel (email: sshlomo@ee.technion.ac.il)}
\thanks{The work of Y. Liang has been supported by the U.S. National Science Foundation under the grant CCF-1618127. The work of H. V. Poor has been supported by the U.S. National Science Foundation under the grants CNS-1702808 and ECCS-1647198. The work of S. Shamai has been supported by the European Union's Horizon 2020 Research And Innovation Programme, grant agreement no. 694630.}
\thanks{This paper was presented (in part) in \cite{Chong17_secrecy}.} %
}

\begin{document}
\maketitle \thispagestyle{empty} \pagestyle{plain}

\begin{abstract}
In this paper, the $k$-th order autoregressive moving average (ARMA(k)) Gaussian wiretap channel with noiseless causal feedback is considered, in which an eavesdropper receives noisy observations of the signals in both forward and feedback channels. It is shown that a variant of the generalized \textit{Schalkwijk-Kailath} scheme, a capacity-achieving coding scheme for the feedback Gaussian channel, achieves the same maximum rate for the same channel with the presence of an eavesdropper. Therefore, the secrecy capacity is equal to the feedback capacity without the presence of an eavesdropper for the feedback channel. Furthermore, the results are extended to the additive white Gaussian noise (AWGN) channel with quantized feedback. It is shown that the proposed coding scheme achieves a positive secrecy rate. As the amplitude of the quantization noise decreases to zero, the secrecy rate converges to the capacity of the AWGN channel.
\end{abstract}



%
%
%

\section{Introduction}

It has been more than a half century since the information theorists started to investigate the capacity of feedback Gaussian channels. As the pioneering studies on this topic, Shannon's 1956 paper \cite{shannon56} showed that feedback does not increase the capacity of the memoryless AWGN channel, and Elias \cite{Elias1956} \cite{Elias1967} proposed some simple corresponding feedback coding schemes. Schalkwijk and Kailath \cite{Schalkwijk66} \cite{Schalkwijk66_2} then developed a notable linear feedback coding scheme to achieve the capacity of the feedback AWGN channel. Thereafter, the problem of finding the feedback capacity and the capacity-achieving codes for the memory Gaussian channels (e.g. ARMA(k)) has been extensively studied. Butman \cite{Butman69} \cite{Butman76}, Wolfowitz \cite{Wolfowitz75} and Ozarow \cite{Ozarow_random90} \cite{Ozarow_upper90} extended Schalkwijk's scheme to ARMA(k) Gaussian channels, leading to several valuable upper and lower bounds on the capacity. Motivated by these elegant results/insights, in $1989$ Cover and Pombra \cite{cover89} made a breakthrough on characterizing the $n$-block capacity of the feedback Gaussian channel. In $2010$, Kim \cite{Kim10} provided a characterization (in the form of an infinite dimensional optimization problem) of the capacity of the ARMA(k) feedback Gaussian channel based on Cover-Pombra's $n$-block capacity characterization. Unfortunately, except for the first-order ARMA noise, it is non-trivial to compute the capacity by solving this infinite dimensional optimization. Recently, Gattami \cite{Ather_feedback_capacity} showed that the capacity of the stationary Gaussian noise channel with finite memory can be found by solving a semi-definite programming problem. In addition, Li and Elia \cite{Li2015_control} \cite{Chong_Youla} utilized the control-theoretic tools to compute the capacity of the ARMA(k) feedback Gaussian channel and construct the capacity-achieving feedback codes.

As a natural extension of the above results, the ARMA(k) feedback Gaussian channel with an eavesdropper that has noisy access of channel transmissions is of much interest. Concretely, two fundamental questions can be asked: 1). would the feedback capacity of such a channel decrease subject to the secrecy constraint?  2). what would be the secrecy capacity-achieving codes? In fact, secure communications over feedback channels have attracted a lot of attention over the last decade. Substantial progresses have been made towards understanding this type of channels. In particular, although the feedback may not increase the capacity of open-loop AWGN channels, \cite{Maurer93, Ahlswede93, Ahlswede2006, poor_07_noisy_fb, Lai08, Ekrem2008, Gunduz2008, Ardestanizadeh2009,Dai_15_entropy, dai_noiseless_feedback} showed that feedback can increase the secrecy capacity by sharing a secret key between legitimate users or roiling the codewords from eavesdroppers' observations. For instance, \cite{Maurer93} and \cite{Ahlswede93} showed the achievement of a positive secrecy rate by using noiseless feedback even when the secrecy capacity of the feed-forward channel is zero. Furthermore, \cite{Bassi2015} presented an achievable scheme for the wiretap channel with generalized feedback, which is a generalization and unification of several relevant previous results in the literature. Very recently, \cite{dai_17_fb_coding} proposed an improved feedback coding scheme for the wiretap channel with noiseless feedback, which was shown to outperform the existing ones in the literature.

On the other hand, the multiple-access (MA) wiretap channel with feedback has also been studied in recent a few years. \cite{poor_07_ma} derived achievable secrecy rate regions for discrete memoryless Gaussian channels where two trusted users send independent confidential messages to an intended receiver in the presence of a passive eavesdropper. \cite{dai_17} developed inner and outer bounds on the secrecy rate region for the MA wiretap channel with noiseless feedback. In a more general setting where users have multiple-input multiple-output (MIMO) equipments, \cite{poor_09} characterized the pre-log factor of the secrecy rate in the case with the number of antennas at the source being no larger than that at the eavesdropper. \cite{Yang2013} and \cite{Tandon2014} investigated the benefits of state-feedback to increase the secrecy degrees of freedom for the two-user Gaussian MIMO wiretap channel.

However, it is noteworthy that most of the aforementioned results on both the point-to-point and the MA communications only considered \textit{memoryless} wiretap channels. Thus motivated, in this paper we consider a \textit{memory} wiretap channel (i.e. ARMA(k) Gaussian noise channel) with feedback and make two major contributions.
\begin{enumerate}
\item  We propose a variant of the generalized $S$-$K$ scheme, which achieves the feedback capacity of the ARMA(k) Gaussian noise channel without an eavesdropper and show that the proposed coding scheme also achieves the same maximum rate of such a channel with an eavesdropper.
\item We further study the AWGN channel with quantized feedback, which is a more realistic channel model for the feedback link. We show that the proposed coding scheme provides non-trivial positive secrecy rates and achieves the feedback capacity of the AWGN channel as the amplitude of the quantization noise vanishes to zero.
\end{enumerate}

The rest of the paper is organized as follows. In Section \ref{sec:sys_model} and \ref{sec:prelim}, we introduce the system model and the preliminary results, respectively. Section \ref{sec::main results} presents the main results of our paper. Section \ref{sec:proofs} provides the technical proofs. Finally, in Section \ref{sec:conclusion}, we conclude the paper and outline possible research venues down the road.

{\bf Notation:} Uppercase and the corresponding lowercase letters $(e.g., Y,Z,y,z)$ denote random variables and their realizations, respectively. We use $\log$ to denote the logarithm with base $2$, and $0\log0=0$. We use $\bf{x}'$ to denote the transpose of a vector or matrix $\bf{x}$. The symbol $C^\infty_{[a,b]}$ denotes the set of bounded continuous functions on $[a,b]$.  $\mathcal{RH}_2$ denotes the set of stable and proper rational filters in Hardy space $\mathcal{H}_2$, and $\mathcal{L}_2$ denotes the 2-norm Lebesgue space.

\section{System Model}{\label{sec:sys_model}}
\indent In this section, we present the mathematical system model. First of all, we consider a discrete-time Gaussian channel with noiseless feedback (Fig. 1).
The additive Gaussian channel is modeled as
\begin{equation}\label{model: forward channel}
y(k)=u(k)+w(k), \qquad k=1,2,\cdots,
\end{equation}
where the Gaussian noise $\lbrace w(k) \rbrace_{k=1}^{\infty}$ is assumed to be stationary with power spectrum density $\mathbb{S}_{w}(e^{j\theta})$ for $\forall \theta\in [-\pi,\pi)$. Unless the contrary is explicitly stated, ``stationary'' without specification refers to stationary in wide sense. Moreover, we assume the power spectral density satisfies the \textit{Paley-Wiener} condition
$$\frac{1}{2\pi}\int_{-\pi}^{\pi}|\log \mathbb{S}_{w}(e^{j\theta})|d\theta< \infty.$$

\begin{figure}
\begin{center}
\includegraphics[scale=0.4]{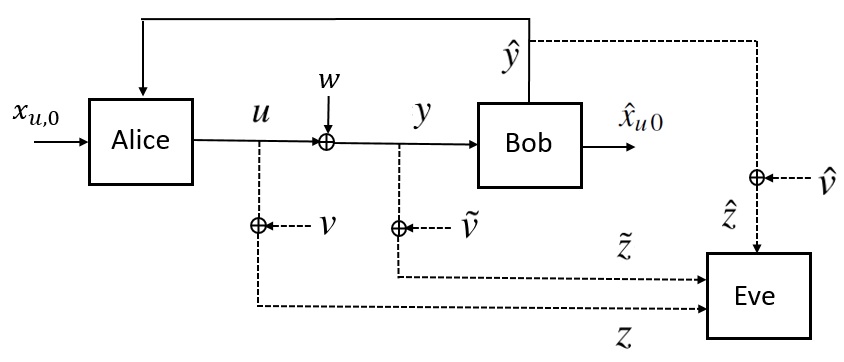}
\caption{ARMA(k) Gaussian wiretap channel with feedback.}
\label{fig:problemModel}
\end{center}
\end{figure}

\begin{assumption} (\textit{ARMA(k) Gaussian Channel})\label{LTIFD.ass}
In this paper, noise $w$ is assumed as the output of a finite-dimensional linear time invariance (LTI) minimum-phase stable system $\mathbb{H} \in \mathcal{RH}_2$, driven by the white Gaussian noise with zero mean and unit variance.
The power spectral density (PSD) of $w$ is colored (nonwhite)\footnote{Both the (Shannon) capacity \cite{shannon56} and the secrecy capacity  \cite{Gunduz2008} for the AWGN channel with noiseless causal feedback have been known.}, bounded away from zero and has a canonical spectral factorization given by  ${\mathbb S}_w(e^{j\theta})= |\mathbb{H}(e^{j\theta})|^2$.
\end{assumption}
As shown in Fig.\ref{fig:problemModel}, the feedback wiretap channel of interest includes a forward channel from Alice to Bob as described by (\ref{model: forward channel}), a causal noiseless feedback $\hat{y}$ from Bob to Alice, and three noisy observation channels to the eavesdropper Eve. Note that a classical wiretap channel model can be recovered if the eavesdropper's channel input $u$ and $\hat{y}$ are removed. In this paper, we assume that the eavesdropper is powerful and can access three inputs \footnote{Note that in \cite{Chong17_secrecy} only the access to channel input $u$ and channel output $y$ is considered. Based on the generalized results in this paper, however, the results in \cite{Chong17_secrecy} with three noisy observation channels to the eavesdropper as assumed in this paper still hold.}. The noisy wiretap channels are modeled as
\begin{equation*}
\begin{split}
z(k) =& u(k) + v(k), \\
\tilde{z}(k) =& y(k) + \tilde{v}(k),  \\
\hat{z}(k) =& \hat{y}(k) + \hat{v}(k), \qquad k= 1,2,\cdots. \\
\end{split}
\end{equation*}
The additive noises $v$, $\tilde{v}$  and $\hat{v}$ are assumed to be arbitrarily finite-memory channels, i.e.,
\begin{equation}\label{model: wiretap channel}
\begin{split}
&p(v(k)|v^{k-1}_1) = p(v(k)|v^{k-1}_{k-d}),\qquad k\geq d,\\
&p(\tilde{v}(k)|\tilde{v}^{k-1}_1) = p(\tilde{v}(k)|\tilde{v}^{k-1}_{k-\tilde{d}}), \qquad k\geq \tilde{d},\\
&p(\hat{v}(k)|\hat{v}^{k-1}_1) = p(\hat{v}(k)|\hat{v}^{k-1}_{k-\hat{d}}), \qquad k\geq \hat{d},\\
\end{split}
\end{equation}
where $d$, $\tilde{d}$ and $\hat{d}$ respectively represent the size of finite memories and the notation $v^b_a$ represents a sequence  $\lbrace v_a, v_{a+1},\cdots, v_{b}\rbrace$ in a compact form. In this paper, we assume these noises have strictly positive and bounded variance for all $k$. But they are not necessarily uncorrelated.

We specify a sequence of $(n,2^{nR_s})$ channel codes with an achievable secrecy rate $R_{s}$ as follows. We denote the message index by $x_{u\,0}$, which is uniformly distributed over the set $\lbrace 1,2,3,\cdots,2^{nR_s}\rbrace$. The encoding process $u_i(x_{u\,0},\hat{y}^{i-1})$ at Alice satisfies the average transmit power constraint $P$, where $\hat{y}^{i-1}=\lbrace \hat{y}_0, \hat{y}_1,\cdots, \hat{y}_{i-1}\rbrace$ ($\hat{y}_0 = \emptyset$) for $i = 1,2,\cdots,n,$  and $u_1(x_{u\,0} , \hat{y}^0) = u_1(x_{u\,0})$. Bob decodes the message to be $\hat{x}_{u\,0}$ by following a decoding function $g$ : $y^n\rightarrow \lbrace 1,2,\cdots,2^{nR_s}\rbrace$ with an error probability satisfying
$P_e^{(n)}=\frac{1}{2^{nR_s}}\sum_{x_{u\,0}=1}^{2^{nR_s}} p({x}_{u\,0}\neq g(y^n)|x_{u\,0})\leq \epsilon_n$,
where $\lim_{n\rightarrow\infty}\epsilon_n=0$. Meanwhile, the information received by Eve should asymptotically vanish, i.e., $\lim_{n\rightarrow\infty} \frac{1}{n}I(x_{u\,0};z_1^n, \tilde{z}_1^n, \hat{z}_1^n )=0$. The objective of secure communications is to send $x_{u\,0}$ to Bob at an as high rate $R_s$ as possible. The secrecy capacity $C_{sc}$ is defined as the supremium of all achievable rates $R_s$. Mathematically,
\begin{equation}
\begin{split}
C_{sc} =& \sup_{\text{feasible coding schemes}} R_s\\
\text{s.t.} & \lim_{n\rightarrow\infty} \frac{1}{n}I(x_{u\,0};z_1^n, \tilde{z}_1^n, \hat{z}_1^n )=0,\\
\end{split}
\end{equation}
where the ``feasible coding schemes'' refer to all feedback codes satisfying the secrecy requirements and the power constraint. Note that the feedback capacity (without the secrecy constraint) from Alice to Bob, denoted as $C_{fb}$, can be recovered by removing the secrecy constraint. This implies $C_{sc}\leq C_{fb}$.

\section{Preliminaries of feedback capacity and capacity-achieving coding scheme}{\label{sec:prelim}}
In this section, we review the characterization of the feedback capacity $C_{fb}$ and then propose a variant of the generalized $S$-$K$ scheme, which is a $C_{fb}$-achieving feedback codes without the presence of an eavesdropper. The materials here are useful for us to further investigate the channel model with an eavesdropper.

\subsection{Feedback Capacity $C_{fb}$ Revisited} \label{sec:computation_C_fb}
Firstly, we present the feedback capacity characterization for the Gaussian channel under Assumption \ref{LTIFD.ass}. As proved in \cite{Kim10}, the feedback capacity from Alice to Bob for such a channel with the average power budget $P$ can be characterized by
\begin{equation}
\begin{split}
C_{fb}=&\max_{\mathbb{Q}}\frac{1}{2\pi}\int_{-\pi}^{\pi}\log |1+\mathbb{Q}(e^{i\theta})|d\theta,\\
s.t. \quad  &\frac{1}{2\pi}\int_{-\pi}^{\pi}|\mathbb{Q}(e^{j\theta})|^2\mathbb{S}_w(e^{j\theta})d\theta\leq P,\\
& \mathbb{Q} \in \mathcal{RH}_2\, \text{is strictly causal}.\\
\end{split}
\label{capacity_short01}
\end{equation}
\begin{remark} \label{rem:unit_zero}
Under Assumption  \ref{LTIFD.ass}, the optimal $\mathbb{Q}$ has no zeros on the unit circle (Proposition 5.1 (ii)  in \cite{Kim10}).
\end{remark}

When specified to ARMA(1) channel, the above characterization can be simplified and has a closed-form solution.
\begin{lemma} \label{lem01}(\textit{Theorem 5.3 in \cite{Kim10}})
For the ARMA(1) feedback Gaussian channel with the average channel input power budget $P>0$, the feedback capacity $C_{fb}$ is given by
\begin{equation*}
C_{fb} = - \log x_0,
\end{equation*}
where $x_0$ is the unique positive root of the fourth-order polynomial in $x$
$$ Px^2 = \frac{(1-x^2)(1+\sigma\alpha x)^2}{(1+\sigma\beta x)},$$
and
$$ \sigma = sign(\beta - \alpha).$$
Furthermore, an optimal $\mathbb{Q}$ to achieve $C_{fb}$ is given by
\begin{equation} \label{equ:optimalQ}
\mathbb{Q}(e^{i\theta}) = -P\sigma x_0 \frac{1+\beta\alpha x_0}{1+\alpha \sigma x_0} \frac{ e^{i\theta} (1+\beta e^{i\theta})}{(1+\alpha e^{i\theta})(1-\sigma x_0 e^{i\theta})}.
\end{equation}

\end{lemma}

Since the optimization in (\ref{capacity_short01}) has infinite dimensional search space, except for the ARMA(1) Gaussian channels, neither the analytical nor the numerical solutions to $\mathbb{Q}(e^{j\theta})$ was known in the literature. One recent result in \cite{Ather_feedback_capacity} casted the above optimization into a semi-definite programming and then used the convex tools to compute $C_{fb}$. In addition, \cite{Li2015_control} \cite{Chong_Youla} provided a numerical approach to compute $C_{fb}$ and explicitly construct the optimal $\mathbb{Q}(e^{j\theta})$, which can be efficiently solved by the standard convex optimization tools. We refer the interested reader to \cite{Li2015_control} \cite{Chong_Youla} for details. To make our paper self-contained, we re-state the main results in \cite{Li2015_control} \cite{Chong_Youla} as follows.


\begin{proposition} (\textit{Lemma 4, \cite{Chong_Youla}})
Consider a non-white additive noise $w$ under Assumption \ref{LTIFD.ass} in the forward channel and let $\mathbb{Q}(e^{j\theta}) = a(\theta)+jb(\theta)$. The $h$-upper-bound on the feedback capacity $C_{fb}$, denoted by $C_{fb}(h)$, can be characterized by
\begin{equation}
\begin{split}
C_{fb}(h)=&\max_{\Gamma}\frac{1}{4\pi}\int_{-\pi}^{\pi}\log ((1+a(\theta))^2+ b(\theta)^2 )d\theta\\
s.t. \quad &\frac{1}{2\pi}\int_{-\pi}^{\pi}\left(a^2(\theta)+b^2(\theta)\right) S_w(\theta)d\theta\leq P,\\
&\int_{-\pi}^{\pi} a(\theta)\cos(n\theta) d\theta + \int_{-\pi}^{\pi} b(\theta) \sin(n\theta)d\theta = 0 \qquad \text{(strict causality constraint)}\\
& \quad n = 0, 1,2,\cdots,h, \\
\end{split}
\label{formula_stationaryGuassian_upperbound_equi}
\end{equation}
where the maximum is taken over a functional set $\Gamma$ defined as
\begin{equation}
\begin{split}
\Gamma =& \lbrace a(\theta), b(\theta): [-\pi, \pi] \rightarrow \mathbb{R} \quad | \quad a(\theta), b(\theta)\in \mathcal{L}_2\rbrace.\\
\end{split}
\label{def_set_Gamma}
\end{equation}
\label{lemma:symmetric_filter}
Furthermore, we have $C_{fb}(h) \geq C_{fb}(h+1), C_{fb}(h) \geq C_{fb}$ for any $h\geq 0$, and
$$C_{fb} = \lim_{h\rightarrow \infty} C_{fb}(h).$$

\end{proposition}

\begin{remark}
With a bit abuse of notation, we use $S_w(\theta)$ instead of $\mathbb{S}_w(e^{i\theta})$ for simplicity. The basic idea to obtain the characterization (\ref{formula_stationaryGuassian_upperbound_equi}) is that the strict causality constraint (\ref{capacity_short01}) can be equivalently imposed on (the infinite number of) the non-positive index coefficients of the inverse Fourier transform of $\mathbb{Q}(e^{i\theta})$ by setting them to zeros. The upper bound follows as we herein only impose a finite number of the Fourier coefficients with non-positive index $-h, -h+1, \cdots, -1,0$ to be zero.
\end{remark}

Notice that $C_{fb}(h)$ still turns out to be a semi-infinite dimensional problem, in which there are finite number of constraints but infinite number of feasible solutions in $\mathcal{L}_2$. However, one important result, as shown below, establishes the strong Lagrangian dual of $C_{fb}(h)$. In other words, there is no duality gap between the semi-infinite dimensional primal problem and its finite dimensional dual problem, a fact that provides a convex optimization approach to compute $C_{fb}(h)$.
\begin{theorem}(\textit{Theorem 2 in \cite{Li2015_control}})\label{strongdual.thm}
Under Assumption \ref{LTIFD.ass}, let
$$
\begin{array}{l}
A(\theta) =[\cos(\theta),\cos(2\theta), \cdots,\cos(h\theta)]',\\
B(\theta) =[\sin(\theta),\sin(2\theta), \cdots,\sin(h\theta)]'.
\end{array}
$$
For $\lambda > 0$, $\eta\in \mathbb{R}^h$, and $\eta_0\in \mathbb{R}$, define
$$
r^2(\theta) =(2\lambda S_{w}(\theta)+\eta'A(\theta)+\eta_0)^2+(\eta'B(\theta))^2.
$$
Then, the following statements hold.\\
a) The Lagrangian dual of (\ref{formula_stationaryGuassian_upperbound_equi}) is given by
\begin{equation}\label{dual.eq}
(D): \mu_h = - \max_{\lambda > 0,\eta\in \mathbb{R}^{h},\eta_0\in \mathbb{R}}g(\lambda,\eta,\eta_0),
\end{equation}
where
\begin{equation}\label{thm.dual.eq}
\begin{split}
&g(\lambda,\eta,\eta_0) = \displaystyle\frac{1}{2\pi}\int_{-\pi}^\pi \left[\frac{1}{2}\log(2\lambda S_{w}(\theta)-\nu(\theta))-\frac{r^2(\theta)}{2\nu(\theta)}+\lambda S_{w}(\theta)\right]d\theta -\lambda P+\eta_0+\frac{1}{2},\\
\end{split}
\end{equation}
with
\begin{equation}\label{optimal.nu}
\nu(\theta)=\frac{-r^2(\theta)+\sqrt{r^4(\theta)+8\lambda S_{w}(\theta)r^2(\theta)}}{2}.
\end{equation}
\\
b) The dual problem (D) in (\ref{dual.eq}) is equivalent to the following convex optimization problem
\begin{equation}\label{dual2.eq}
\mu_h = - \max_{\begin{array}{l}\lambda > 0,\eta\in \mathbb{R}^{h},\eta_0\in \mathbb{R}\\ \nu(\theta)\geq 0 \in C^\infty_{[-\pi,\pi]}\end{array}}
\tilde{g}(\lambda,\eta,\eta_0,\nu(\theta)),
\end{equation}
where
\begin{equation}\label{thm.dual2.eq}
\begin{split}
&\tilde{g}(\lambda,\eta,\eta_0,\nu(\theta)) = \displaystyle\frac{1}{2\pi}\int_{-\pi}^\pi \left[\frac{1}{2}\log(2\lambda S_{w}(\theta)-\nu(\theta))-\frac{r^2(\theta)}{2\nu(\theta)}+\lambda S_{w}(\theta)\right]d\theta -\lambda P+\eta_0+\frac{1}{2},\\
\end{split}
\end{equation}
and the optimal $\nu(\theta)$ is characterized by (\ref{optimal.nu}).\\
c) Furthermore, $C_{fb}(h)=\mu_h$, and an optimal strictly casual filter $\mathbb{Q}_h(e^{j\theta}) = a(\theta)+jb(\theta)$ for $C_{fb}(h)$ exists and is characterized by
\begin{equation}\label{thm_xy_dual}
\begin{split}
a(\theta)=&\frac{2\lambda S_{w}(\theta)+\eta'A(\theta)+\eta_0}{\nu(\theta)}-1 \quad a.e.,\\
b(\theta)=&\frac{\eta'B(\theta)}{\nu(\theta)}  \quad a.e.,\\
\end{split}
\end{equation}
\label{lemma.strong.dual}
where $(\lambda, \eta', \eta_0, \nu(\theta))$ are obtained by solving (\ref{dual.eq}) (or (\ref{dual2.eq})).
\end{theorem}

Note that the above optimization is not easily computable due to the integral in the objective function. A natural approach is to approximate the integral with a finite sum by discretizing $\theta$. We apply such discretization to (\ref{dual2.eq}) (with spacing $\frac{\pi}{m}$) and introduce the following finite dimensional convex problem. Given $m$,  consider
\begin{equation}{\label{opt_upperbound_approximate}}
\mu_h(m)=-\max_{\lambda > 0,\eta,\eta_0, \nu_i\geq 0}\tilde{g}_m(\lambda,\eta,\eta_0,\nu_i)
\end{equation}
where
\begin{equation*}
\begin{split}
&\tilde{g}_m(\lambda,\eta,\eta_0,\nu_i) =  \frac{1}{2m}\sum_{i=1}^{2m}\left(\frac{1}{2}\log(2\lambda S_{w}(\theta_i)-\nu_i)+\lambda S_{w}(\theta_i)-\frac{r^2(\theta_i)}{2\nu_i} -\lambda P+\eta_0+\frac{1}{2}\right),\\
\end{split}
\end{equation*}
and $\theta_i = -\pi+ \frac{\pi}{m}(i-1)$. \\
It is argued in \cite{Chong_Youla} that $\mu_h(m)$ can be an arbitrarily well approximation of $\mu_h$, i.e., $\lim_{m\rightarrow \infty} \mu_h(m) = \mu_h$.
Notice that the optimization (\ref{opt_upperbound_approximate}) is in a simple convex form. In particular, the $\log$ of an affine function is concave. The term $\frac{r^2(\theta_i)}{\nu_i}$ is quadratic (composed with an affine function of the variables)  over linear function, therefore convex.  Thus,  (\ref{opt_upperbound_approximate}) can be efficiently  solved  with the standard convex optimization tools, e.g. CVX, a package for specifying and solving convex programs\cite{cvx01, cvx02}.

In what follows, we describe, given an optimal $\mathbb{Q}$ in (\ref{capacity_short01}), how to construct an implementable coding scheme that achieves the feedback capacity from Alice to Bob.

\begin{figure*}
\begin{center}
\includegraphics[scale=0.45]{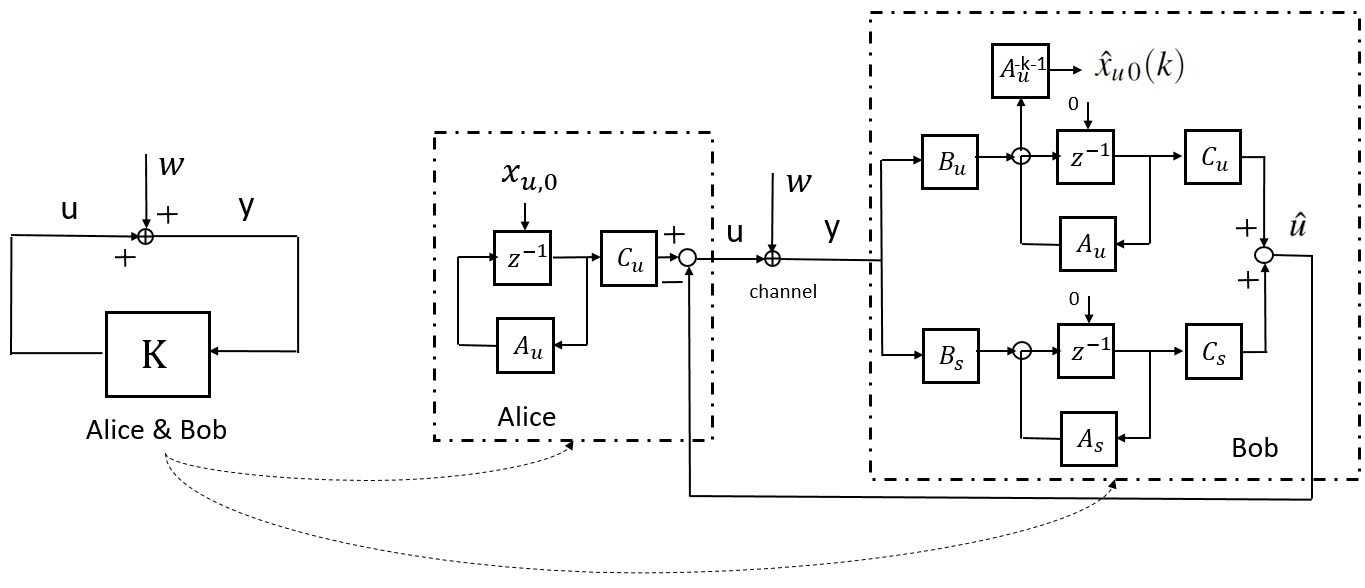}
\caption{Representation I: Decomposition of filter $\mathbb{K}$ into the feedback encoder (Alice) and decoder (Bob). The eavesdropper channels are not included. $z$-transform is used to represent the dynamics of LTI systems.}
\label{fig:codingStructure_elia}
\end{center}
\end{figure*}
%


\subsection{$C_{fb}$-achieving Feedback Coding Scheme}\label{sec: feedback capacity}
First of all, once an optimal $\mathbb{Q}$ is found for the above optimization, we construct a feedback filter $\mathbb{K}=-\mathbb{Q}(1+\mathbb{Q})^{-1}$ stabilizing the channel within the prescribed input average power budget (see \cite{Li2015_control} for proofs). Next, based on the transfer function $\mathbb{K}$, we construct an explicit feedback coding scheme as follows, which is deterministic (time-invariant) and has doubly exponential decaying decoding error probability.

We first present controller $\mathbb{K}$ as a LTI single-input-single-output (SISO) finite-dimensional discrete-time unstable system with the following state-space model:
\begin{equation}
\begin{split}
\mathbb{K}: \qquad \begin{bmatrix} x_s(k+1) \\ x_u(k+1)\end{bmatrix} &= \begin{bmatrix} A_s & 0 \\ 0 & A_u\end{bmatrix} \begin{bmatrix} x_s(k) \\ x_u(k)\end{bmatrix} + \begin{bmatrix} B_s \\ B_u\end{bmatrix}y(k)\\
u(k) &= \begin{bmatrix} C_s & C_u\end{bmatrix} \begin{bmatrix} x_s(k) \\ x_u(k)\end{bmatrix}.\\
\end{split}
\label{codingScheme_SS}
\end{equation}
Based on Remark \ref{rem:unit_zero}, we assume that the eigenvalues of $A_u$ are strictly outside the unit disc while the eigenvalues of $A_s$ are strictly inside the unit disc. Without loss of generality, we assume that $A_s$ and $A_u$ are in Jordan form. Assume $A_u$ has $m$ eigenvalues, denoted by $\lambda_i(A_u), i = 1,2,\cdots, m$. Next, we review the coding scheme in \cite{Chong_Youla} \cite{Elia2004} (Fig. \ref{fig:codingStructure_elia}), which decompose ${\mathbb K}$ into an encoder (Alice) and a decoder (Bob) with an estimated signal fed back to the encoder via the noiseless feedback channel.

{\textit {{\bf Representation I}: Decoder-estimation-based Feedback Coding Scheme}}\cite{Chong_Youla} \cite{Elia2004} (Fig. \ref{fig:codingStructure_elia}).\\

{\bf{Decoder:}} The decoder runs ${\mathbb K}$ driven by the channel output $y$.
$$
\begin{array}{ccl}
x_s(k+1)&=&A_sx_s(k)+B_sy(k),\; x_{s}(0)=0,\\
\hat{x}_u(k+1)&=&A_u\hat{x}_u(k)+B_uy(k),\; \hat{x}_{u}(0)=0.
\end{array}
$$
It produces two signals:
an estimate of the initial condition of the encoder
$$
\hat{x}_{u\,0}(k)=A_u^{-k-1}\hat{x}_u(k+1),
$$
and a feedback signal
$$
\hat{u}(k) = \begin{bmatrix} C_s & C_u\end{bmatrix} \begin{bmatrix} x_s(k) \\ \hat{x}_u(k)\end{bmatrix}.\\
$$

{\bf Encoder:} The encoder runs the following dynamics
$$
\begin{array}{rcl}
\tilde{x}_u(k+1)&=&A_u\tilde{x}_u(k),\;\tilde{x}_u(0)=x_{u\,0},\\
\tilde{u}_u(k)&=&C_u\tilde{x}_u(k).
\end{array}
$$
It receives $\hat{u}$ and produces the channel input
$$
u(k)=\tilde{u}_u(k)-\hat{u}(k).
$$

As will be seen later, the above coding scheme cannot be directly applied to the feedback ARMA(k) Gaussian channel with an eavesdropper. Therefore, we next propose an equivalent representation of the above decoder-estimation-based feedback coding scheme. It will be proved that a variant of this new representation achieves the same maximum rate of the same channel without the eavesdropper.\\

{\textit {{\bf Representation II}: Channel-output-based Feedback Coding Scheme (Fig. \ref{fig:codingStructure_est_feedback})}}.\\

{\bf{Decoder:}} The decoder runs ${\mathbb K}$ driven by the channel output $y$.
$$
\begin{array}{ccl}
\hat{x}_u(k+1)&=&A_u\hat{x}_u(k)+B_uy(k),\; \hat{x}_{u}(0)=0.
\end{array}
$$

It only produces an estimate of the initial condition of the encoder
$$
\hat{x}_{u\,0}(k)=A_u^{-k-1}\hat{x}_u(k+1).
$$

{\bf Encoder:} The encoder runs the following dynamics driven by the initial state, i.e., the message,
$$
\begin{array}{rcl}
\tilde{x}_u(k+1)&=&A_u\tilde{x}_u(k),\;\tilde{x}_u(0)=x_{u\,0},\\
\tilde{u}_u(k)&=&C_u\tilde{x}_u(k).
\end{array}
$$
It receives $y$ and runs dynamics driven by the received feedback $y$,
$$
\begin{array}{ccl}
x_s(k+1)&=&A_sx_s(k)+B_sy(k),\; x_{s}(0)=0\\
\hat{x}_u(k+1)&=&A_u\hat{x}_u(k)+B_uy(k),\; \hat{x}_{u}(0)=0,
\end{array}
$$
and produces a signal
$$
\hat{u}(k) = \begin{bmatrix} C_s & C_u\end{bmatrix} \begin{bmatrix} x_s(k) \\ \hat{x}_u(k)\end{bmatrix}.\\
$$

Then, the encoder produces the channel input
$$
u(k)=\tilde{u}_u(k)-\hat{u}(k).
$$

\begin{figure}
\begin{center}
\includegraphics[scale=0.45]{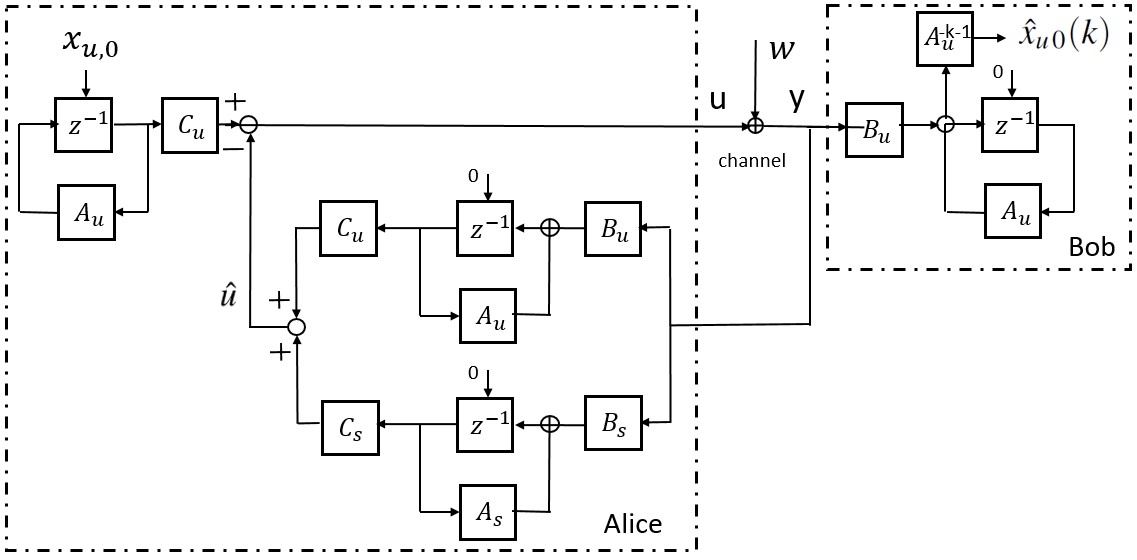}
\caption{Representation II: Decomposition of filter $\mathbb{K}$ into the feedback encoder (Alice) and decoder (Bob). The eavesdropper channels are not included.}
\label{fig:codingStructure_est_feedback}
\end{center}
\end{figure}

By comparing the two representations, we see that the only difference comes from the feedback signal. In \textit{Representation I}, the feedback signal $\hat{u}$ is generated by the decoder, while in \textit{Representation II}, the feedback signal is simply the raw channel output. The equivalence can be directly verified by comparing the channel inputs $u$ (encoder) and the estimate of the message $\hat{x}_{u\,0}$ (decoder) of the two representations.

\begin{proposition}\label{prop:equivalence repr}
For a given message $x_{u\,0}$ and a sequence of additive noise $w_1^k$ ($k\geq 1$), the {\textit{ Representation I}} and {\textit{ Representation II}} of the proposed coding scheme produce identical channel input $u(k)$ and message estimate $\hat{x}_{u\,0}(k)$ for $\forall k$.
\end{proposition}

\begin{remark} \label{remark:equivalence} It is important to notice that the ``equivalence'' only holds for such a channel \textit{without} an eavesdropper. This is because in our model the eavesdropper can access the feedback link. In \textit{Representation II}, since the channel output is directly fed back to Alice, the eavesdropper's access to both the channel output and the feedback link has no difference from the access to the channel output only. However, this is clearly not true for  \textit{Representation I}, in which the eavesdropper can extract more useful information from the decoding process (Bob) by accessing the feedback link.
\end{remark}

We next provide some insight on the proposed coding scheme. All the discussions in the rest of this section hold for both representations.

First of all, in the above proposed coding schemes, since the closed loop is stable,  $u(k)$ goes to zero with time if the noise is not present. Given that
\begin{equation}
\begin{split}
u(k)=&\tilde{u}_u(k)-\hat{u}(k)\\
=& C_u(\tilde{x}_u(k)-\hat{x}_u(k)) - C_s x_s(k),\\
\end{split}
\end{equation}
and the system is observable, we have $\hat{x}_u(k)\to -\tilde{x}_u(k)$. Thus, $-\hat{x}_{u\,0}(k)$ is an  estimate at time $k$ of $\tilde{x}_u(0)=x_{u,0}$. In the presence of noise, Theorem 4.3 in \cite{Elia2004} shows that the above coding scheme leads to $\hat{x}_{u\,0}(k-1) \backsim \mathcal{N}(-x_{u,0},A_u^{-k}\mathbb{E}[\hat{x}_u(k)\hat{x}_u(k)'] (A_u^{-k})')$ for large $k$, where $\mathbb{E}[\hat{x}_u(k)\hat{x}_u(k)']$ represents the state covariance matrix. Note that, since the stable closed-loop system is observable and controllable, the matrix $\mathbb{E}[\hat{x}_u(k)\hat{x}_u(k)']$ is positive definite (which converges to a steady state value and is independent from the initial state of the system) and so is $A_u^{-k}\mathbb{E}[\hat{x}_u(k)\hat{x}_u(k)'] (A_u^{-k})'$. Since the eigenvalues of $A_u$ are strictly outside the unit disc, it is straightforward to see that the state covariance matrix converges to zero (e.g. element-wise) as $k$ goes to infinity.

In addition, in the above coding scheme the message index $x_{u0}\in \mathbb{R}^m$ is allocated at the centroid of an unit hypercube in the coordinate system depending on $A_u \in \mathbb{R}^{m\times m}$. We refer interested readers to Theorem 4.3 in \cite{Elia2004} for details. For scalar $A_u \in \mathbb{C}$, the unit hypercube becomes an interval, e.g. $[-0.5,0.5]$. That is, $2^{nR_s}$ messages are represented by the middle points of equally divided $2^{nR_s}$ subintervals within $[-0.5,0.5]$. In the end, we note that the causality of the feedback channel is captured by the one-step delay of the state transition on $x_s$ and $\hat{x}_u$.

The next lemma indicates that the proposed coding scheme $\mathbb{K}$ is $C_{fb}$-achieving. The proof is omitted as it follows directly from Theorem 4.3 and Theorem 4.6 in \cite{Elia2004}.
\begin{proposition}
Consider stationary Gaussian channels in (\ref{capacity_short01}). Given a filter $\mathbb{Q}(e^{i\theta})\in \mathcal{RH}_2$, the proposed coding scheme based on the decomposition of $\mathbb{K} = -\mathbb{Q}(1+\mathbb{Q})^{-1}$ achieves a reliable transmission rate (in the sense of Shannon) at
$$\frac{1}{2\pi}\int_{-\pi}^{\pi}\log |1+\mathbb{Q}(e^{i\theta})|d\theta =\sum_{i=1}^{m} \log|\lambda_i(A_u)| \quad  \textit{bits/channel use},$$
and has double exponential decaying error probability.
\label{thm_capacity_achieving_code}
\end{proposition}

This proposition indicates that the achievable rate of the proposed coding schemes ({\textit{Representation I}} and {\textit{ Representation II}}) is explicitly characterized by the objective function in (\ref{capacity_short01}), leading to the fact that the proposed coding schemes are capacity-achieving if $\mathbb{Q}(e^{i\theta})$ is an optimal solution to (\ref{capacity_short01}). Moreover, this achievable rate is only determined by the unstable eigenvalues of the system, implying that the capacity-achieving coding scheme $\mathbb{K}$ must be unstable although the closed-loop system is stable.

Thus far, we have presented explicit feedback codes (in the state-space representation) that can achieve the feedback capacity. In the next section, we show that a variant of this coding scheme ${\mathbb K}$ can lead to the asymptotic zero leakage of the message to Eve. This implies $C_{sc} = C_{fb}$.
\section{Main Results}\label{sec::main results}
In this section, we first present our results for the ARMA(k) feedback channel with an eavesdropper, and then extend our results to the case with quantization noise in feedback.

\subsection{ARMA(k) Feedback Channel with an Eavesdropper}

We first present our new development on the properties of the feedback coding scheme for ARMA(k) Gaussian channels without the presence of an eavesdropper. We then use these properties to establish our main theorem that characterizes the feedback secrecy capacity and its achieving coding scheme.

The following result shows that, by choosing the particular $m$-step initializations (in the state-space representation) for the proposed coding scheme, the channel inputs ($k\geq m+1$) are only determined by the past additive Gaussian noise $w$, a fact that is vital to guarantee the asymptotic secrecy from Eve.
\begin{proposition}\label{prop_channel_input}
For the proposed coding scheme in Fig. \ref{fig:codingStructure_elia} or Fig. \ref{fig:codingStructure_est_feedback}, assume the first $m$-step channel inputs $u_1^m = A_u^{m+1} x_{u\,0}$ (where $A_u^{m+1}$ refers to matrix $A_u$ with power $m+1$), the estimate message $\hat{x}_{u\,0}(m) = {A_u^{-m-1}}{y_1^m}$ (or equivalently, $\hat{x}_u(m+1) = {y_1^m}$) and $x_s(m+1) = 0$, where $m$ is the number of the eigenvalues of matrix $A_u$. Then the induced channel inputs $u(k)$ for $k\geq m+1$ are only determined by the past Gaussian noise $w_1^{k-1}$.
\end{proposition}

\begin{proof}
See Section \ref{prop:channel_input}.
\end{proof}

\begin{proposition}\label{prop_channel_input_capacity_acheving}
With the initializations defined in Proposition \ref{prop_channel_input}, the coding scheme ${\mathbb K}$ in Fig. \ref{fig:codingStructure_elia} or Fig. \ref{fig:codingStructure_est_feedback} remains to be $C_{fb}$-achieving.
\end{proposition}

\begin{proof}
The proof follows directly from two facts. On the one hand, all these initials have no effect on the average transmission power on channel inputs, which only depend on the steady state of the underlying LTI systems. On the other hand, these initials do not change the reliable transmission rate of the coding scheme ${\mathbb K}$, which is defined in an asymptotic manner and only determined by the unstable eigenvalues in $A_u$ (Proposition \ref{thm_capacity_achieving_code}).
\end{proof}

This result provides a variant of the coding scheme ${\mathbb K}$, in which the first $m$ steps ($k=1, 2, \cdots, m$) of the channel input $u$ are different from others ($k\geq m+1$) and the states $\hat{x}_u$ and $x_s(m+1)$ of the coding dynamics are particularly initialized at $m+1$. By doing so, after the initialization phase ($k\leq m$), Alice only transmits the statistical learning information of ``noise'' rather than the message for $k\geq m+1$ in the forward channel, such that if only having access to the forward channel (which is essentially the case for \textit{ Representation II}) Eve cannot learn innovative information of the message through its noisy wiretap channels except in the first $m$-steps.


\begin{theorem}\label{thm_main}
Consider the ARMA(k) Gaussian wiretap channel with feedback (Fig. \ref{fig:problemModel}) under the average channel input power constraint $P>0$. Then,
\begin{enumerate}
\item the feedback secrecy capacity equals the feedback (Shannon) capacity, i.e., $C_{sc} = C_{fb}$, where $C_{fb}$ is obtained from Section \ref{sec:computation_C_fb};
\item the feedback secrecy capacity is achieved by the $C_{fb}$-achieving feedback coding scheme ${\mathbb K}$ (\textit{ Representation II}) with $u_1^m = A_u^{m+1} x_{u\,0}$, $\hat{x}_{u\,0}(m) = {A_u^{-m-1}}{y_1^m}$ (or equivalently, $\hat{x}_u(m+1) = {y_1^m}$), and $x_s(m+1) = 0$.
\end{enumerate}
\end{theorem}

\begin{proof}
See Section \ref{thm: main}.
\end{proof}

This theorem shows that there exists a feedback coding scheme such that the secrecy requirement can be achieved without loss of the communication rate of the legitimate users. In addition, Section \ref{sec: feedback capacity} provides such a feedback coding scheme that achieves the feedback secrecy capacity. In particular, a $C_{sc}$-achieving feedback code can be constructed from the optimal $\mathbb{Q}$ in (\ref{capacity_short01}) by following the procedures in Section \ref{sec: feedback capacity} (\textit{Representation II}) with the initializations defined in Proposition \ref{prop_channel_input}. The next corollary shows that the well-known $S$-$K$ scheme \cite{Schalkwijk66_2} is a special case of our proposed coding scheme.

\begin{corollary} \label{col:SK_robustness}
Consider the AWGN wiretap channel with feedback (Fig. \ref{fig:problemModel}) under the average channel input power constraint $P>0$. Assume that the additive noise $w$ has zero mean and variance $\sigma_w^2>0$. Then the proposed coding scheme ${\mathbb K}$ (\textit{Representation II}) with
$A_u = \sqrt{\frac{P+\sigma_w^2}{\sigma_w^2}}, \quad B_u = -\frac{\sqrt{A_u^2-1}}{A_u}, \quad C_u = -\sqrt{A_u^2-1},$
and $A_s=B_s = C_s = 0$ becomes the original \textit{S-K} scheme, and achieves the secrecy capacity
$C_{sf} = C_{fb} = \frac{1}{2}\log(1+\frac{P}{\sigma_w^2}).$
\end{corollary}

 \begin{proof}
See Section \ref{col: SK}.
\end{proof}

This corollary recovers Theorem 5.1 in \cite{Gunduz2008}, showing that the well-known \textit{S-K} scheme not only achieves the feedback capacity but also automatically provides the secrecy from the eavesdropper.

\begin{remark}
As discussed in Remark \ref{remark:equivalence}, if the eavesdropper has no access to the feedback link, then the above results hold for the coding scheme in \textit{Representation I} as well. That is, the feedback secrecy capacity $C_{sc}$ is achieved by the $C_{fb}$-achieving feedback coding scheme ${\mathbb K}$ (\textit{Representation I}) with the selected initializations. The proof of this argument follows directly from the proof of Theorem \ref{thm_main} by deleting the eavesdropper's access to the feedback link. Therefore, we conclude that the channel-output-based feedback coding scheme \textit{Representation II} is highly necessary to remove the advantage of the eavesdropper's access to the feedback link. If no such access, the ``equivalence'' between \textit{Representation I} and \textit{Representation II} holds for channels with the eavesdropper.
\end{remark}
\subsection{Feedback with Quantization Noise }\label{sec:quantized}
In this section, we extend our result to Gaussian channels with quantized feedback. It is noteworthy that the capacity of colored Gaussian channels with noisy feedback remains an open problem \cite{chong_isit11_capacity} \cite{Chong11_allerton_upperbound} \cite{chong.thesis}, even when simplified to the quantization feedback. Therefore, in this paper, as an initial step towards the secrecy capacity of noisy feedback Gaussian channels, we focus on AWGN channels with quantized feedback. In \cite{Martins08}, the authors presented a linear coding scheme featuring a positive information rate and a positive error exponent for AWGN channels with feedback corrupted by quantization or bounded noise. In what follows, we show that our proposed \textit{linear} coding scheme \textit{Representation II}, when specified to the AWGN channel with quantized feedback, converges to the scheme in \cite{Martins08} and, more importantly, leads to a positive secrecy rate. Furthermore, this achievable secrecy rate converges to the capacity of the AWGN channel as the amplitude of the quantization noise decreases to zero.

Firstly, we define a memoryless uniform quantizer with sensitivity $\sigma_q$ as follows \cite{Martins08}.
\begin{definition}\label{def:quantizer}
Given a real parameter $\sigma_q>0$, a uniform quantizer with sensitivity $\sigma_q$ is a function $\Phi_{\sigma_q}$: $\mathbb{R}\rightarrow \mathbb{R}$ defined as
$$\Phi_{\sigma_q}(y) = 2 \sigma_q \lfloor{\frac{y+\sigma_q}{2 \sigma_q}}\rfloor,$$
where $\lfloor \cdot \rfloor$ represents the floor function. Then, the quantization error at instant $k$, i.e., the feedback noise, is given by
$$ q(k) = \Phi_{\sigma_q}(y(k)) - y(k).$$
\end{definition}

Notice that, for a given channel output $y(k)$, the quantization noise $q(k)$ can be recovered by the decoder as we assume the decoder knows the quantization rule. In other words, the decoder can get access to both the channel outputs and the feedback noise while the encoder can only get access to the corrupted channel output. On the other hand, note that with the quantized feedback the coding schemes \textit{Representation I} and \textit{Representation II} are not equivalent any more due to the different feedback signals. The \textit{Representation I} may not be applicable here. Therefore, we tailor the coding scheme \textit{Representation II} to the AWGN channel with quantized feedback as follows (Fig. \ref{fig:codingStructure_quantization}). We first let $A_s = B_s = C_s =0$.

\begin{figure*}
\begin{center}
\includegraphics[scale=0.4]{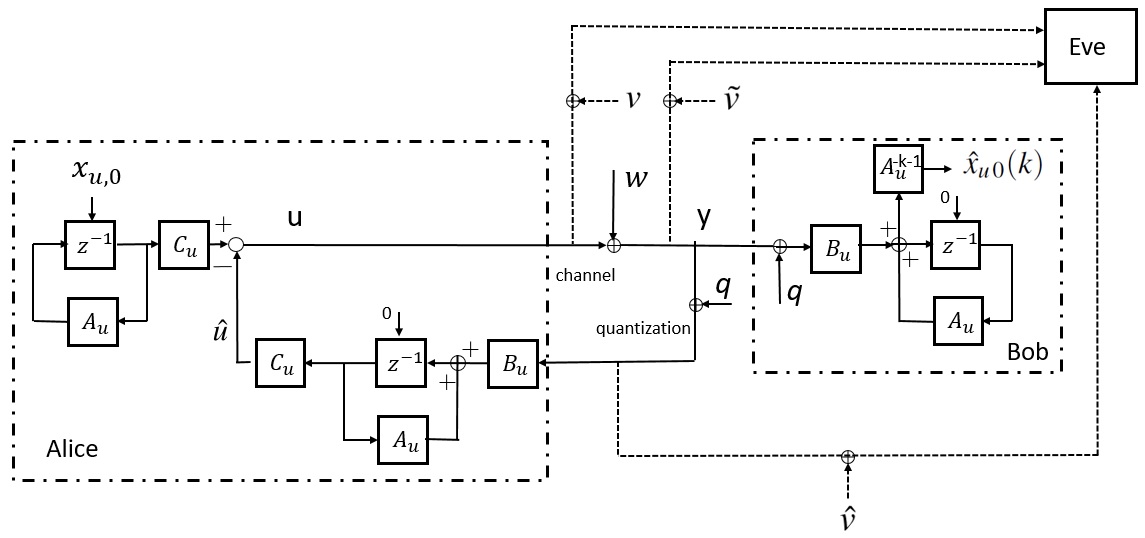}
\caption{Coding Structure for the AWGN channel with quantized feedback. The quantization noise $q$ can be recovered by the decoder to help decoding. }
\label{fig:codingStructure_quantization}
\end{center}
\end{figure*}

{\bf{Decoder:}} The decoder runs ${\mathbb K}$ driven by the sum of the channel output $y$ and the quantization noise $q$.
$$
\begin{array}{ccl}
\hat{x}_u(k+1)&=&A_u\hat{x}_u(k)+B_u (y(k)+q(k)),\; \hat{x}_{u}(0)=0.
\end{array}
$$
It produces an estimate of the initial condition of the encoder
$$
\hat{x}_{u\,0}(k)=A_u^{-k-1}\hat{x}_u(k+1).
$$

{\bf Encoder:} The encoder runs the following dynamics
$$
\begin{array}{rcl}
\tilde{x}_u(k+1)&=&A_u\tilde{x}_u(k),\;\tilde{x}_u(0)=x_{u\,0},\\
\tilde{u}_u(k)&=&C_u\tilde{x}_u(k).
\end{array}
$$
It receives $y+v$ and duplicates the decoding dynamics,
$$
\begin{array}{ccl}
\hat{x}_u(k+1)&=&A_u\hat{x}_u(k)+B_u (y(k)+q(k)),\; \hat{x}_{u}(0)=0,
\end{array}
$$
and produces a signal,
$$
\hat{u}(k) =  C_u \hat{x}_u(k).
$$
Then, it produces the channel input
$$
u(k)=\tilde{u}_u(k)-\hat{u}(k).
$$

We next show that the above coding scheme can achieve a positive secrecy rate, which converges to the AWGN capacity as the feedback noise $\sigma_q$ decreases. The following definition will be used to characterize this secrecy rate.
\begin{definition} \label{def_rv} \cite{Martins08}
For the given positive real parameters $\sigma_w^2, \sigma_q$ and the power constraint $P$, define a parameter $r_q$ as follows.
\begin{enumerate}
\item If $4\sigma_q \leq P$, $r_q$ is the nonnegative real solution of the following equation,
$$ \sigma_w \sqrt{2^{2r_q}-1} = \sqrt{P} - \sigma_q (1+ 2^{r_q}). $$
\item If $4\sigma_q > P$, then $r_q = 0$.
\end{enumerate}
\end{definition}

It is easy to check that $r_q$ satisfies the following three properties \cite{Martins08}
\begin{enumerate}
\item $r_q$ converges to the AWGN capacity as $\sigma_q$ decreases, i.e.,
 $$ \lim_{\sigma_q \rightarrow 0^+} r_q = \frac{1}{2}\log(1+\frac{P}{\sigma_w^2}).$$
\item If $\sigma_q = \frac{\sqrt{P}}{2}$, we have $r_q = 0$.
\item If $P \gg \max\lbrace \sigma_w^2, \sigma_q^2\rbrace$, $r_q \simeq \log(\frac{\sqrt{P}}{\sigma_w+\sigma_q})$. In other words, the ratio of $r_q$ and $\log(\frac{\sqrt{P}}{\sigma_w+\sigma_q})$ converges to $1$ as $P\rightarrow \infty$.
\end{enumerate}

\begin{theorem}\label{thm:quantized}
Consider an AWGN channel with uniformly memoryless quantized feedback defined in Definition \ref{def:quantizer}, where the channel input power constraint is $P>0$ and the noise variance of the AWGN channel and the quantization sensitivity in the feedback link are assumed to be $\sigma_w^2$ and $\sigma_q$, respectively. Assume $u(1) = A_u^{2} x_{u\,0}$, and $\hat{x}_{u\,0}(1) = {A_u^{-2}}(y(1)+q(1))$ (or equivalently, $\hat{x}_u(2) = y(1)+q(1)$). Then, the above proposed coding scheme with $A_u = 2^r, B_u = -1, C_u = A_u-\frac{1}{Au}$ and $A_s = B_s = C_s =0$ achieves a secrecy rate $r$ for all $r < r_q$ ($r_q$ is defined in Definition \ref{def_rv}).
\end{theorem}

\begin{proof}
See Section \ref{thm: quantized}.
\end{proof}

Combined with the aforementioned property (1) on $r_q$, this theorem implies that the achievable feedback secrecy rate of the proposed coding scheme converges to the AWGN capacity as $\sigma_q$ decreases to zero.

\section{Technical Proofs}{\label{sec:proofs}}

In this section, we present the omitted proofs in Sections \ref{sec::main results}.
%

\subsection{Proof of Proposition \ref{prop_channel_input}} \label{prop:channel_input}

Based on the \textit{ Representation I} of the proposed feedback coding scheme, we have
\begin{equation}\label{equ:input_u_orthoganal}
\begin{split}
u(k)=&\tilde{u}_u(k)-\hat{u}(k)\\
=& C_u(\tilde{x}_u(k)-\hat{x}_u(k)) - C_s x_s(k)\\
=& C_u(A_u\tilde{x}_u(k-1)- A_u^{k}\hat{x}_{u\,0}(k-1)) - C_s x_s(k)\\
= & \cdots\\
=& C_u (A_u^{k} x_{u\,0}-  A_u^{k} \hat{x}_{u\,0}(k-1)) - C_s x_s(k)\\
=& C_u A_u^{k} (x_{u\,0}- \hat{x}_{u\,0}(k-1)) - C_s x_s(k).\\
\end{split}
\end{equation}
Next, for $k\geq m+1$ where the initial steps have passed and the signals evolve as described in the coding scheme, it yields
\begin{equation}\label{equ:hat_msg}
\begin{split}
\hat{x}_{u\,0}(k)=&A_u^{-k-1}\hat{x}_u(k+1)\\
=& A_u^{-k-1}(A_u\hat{x}_u(k)+B_u y(k)) \\
=& A_u^{-k}\hat{x}_u(k)+A_u^{-k-1}B_u y(k) \\
=& \hat{x}_{u\,0}(k-1)+A_u^{-k-1}B_u y(k) \\
=& \hat{x}_{u\,0}(k-1)+A_u^{-k-1}B_u (u(k)+w(k))\\
\stackrel{(a)}{=}& \hat{x}_{u\,0}(k-1)+A_u^{-k-1}B_u (C_u A_u^{k} x_{u\,0}- C_uA_u^{k}\hat{x}_{u\,0}(k-1) - C_s x_s(k)+w(k))\\
=& \hat{x}_{u\,0}(k-1) - A_u^{-k-1}B_u C_uA_u^{k}\hat{x}_{u\,0}(k-1)+A_u^{-k-1}B_u C_u A_u^{k} x_{u\,0} + A_u^{-k-1}B_u ( w(k) - C_s x_s(k))+ x_{u\,0} - x_{u\,0}\\
=& (I - A_u^{-k-1}B_u C_uA_u^{k})\hat{x}_{u\,0}(k-1) -(I- A_u^{-k-1}B_u C_uA_u^{k})x_{u\,0} + A_u^{-k-1}B_u (w(k) - C_s x_s(k)) + x_{u\,0}\\
=& (I -  A_u^{-k-1}B_u C_uA_u^{k})(\hat{x}_{u\,0}(k-1) -x_{u\,0}) + A_u^{-k-1}B_u (w(k) - C_s x_s(k)) + x_{u\,0},\\
\end{split}
\end{equation}
where step (a) follows from (\ref{equ:input_u_orthoganal}). Let $\alpha_k = I -  A_u^{-k-1}B_u C_uA_u^{k}$ and $\beta_k = A_u^{-k-1}B_u$. Moving $x_{u\,0}$ to the left side, we have
\begin{equation*}
\begin{split}
&\hat{x}_{u\,0}(k) - x_{u\,0} = \alpha_k(\hat{x}_{u\,0}(k-1) -x_{u\,0}) + \beta_k  (w(k) - C_s x_s(k)).\\
\end{split}
\end{equation*}
By iterating the above equation, for $k\geq m+1$, we obtain
\begin{equation*}
\begin{split}
&\hat{x}_{u\,0}(k) - x_{u\,0}= \prod_{i=m+1}^k \alpha_i (\hat{x}_{u\,0}(m) -x_{u\,0}) + \sum_{i=m+1}^{k} \prod_{j=i+1}^k \alpha_j \beta_i (w(i) - C_s x_s(i)).\\
\end{split}
\end{equation*}
where we assume $\alpha_{k+1} = 1$. Given $u_1^m = A_u^{m+1} x_{u\,0}$ and $\hat{x}_{u\,0}(m) = {A_u^{-m-1}}{y_1^m}$, we have
$$ \hat{x}_{u\,0}(m) = A_u^{-m-1} (u_1^m + w_1^m) = x_{u\,0}+ {A_u^{-m-1}}{w_1^m}.$$
Then, it yields
\begin{equation*}
\begin{split}
&\hat{x}_{u\,0}(k) - x_{u\,0}= \prod_{i=m+1}^k \alpha_i {A_u^{-m-1}}{w_1^m} + \sum_{i=m+1}^{k} \prod_{j=i+1}^k \alpha_j \beta_i (w(i) - C_s x_s(i)).\\
\end{split}
\end{equation*}

Furthermore, given $x_s(m+1) = 0$, from (\ref{equ:input_u_orthoganal}) we have
\begin{equation}
\begin{split}
u(m+1) =& -C_u A_u^{m+1} ( \hat{x}_{u\,0}(m) - x_{u\,0}) - C_s x_s(m+1)\\
 =& -C_u A_u^{m+1} ( \hat{x}_{u\,0}(m) - x_{u\,0}) \\
 =& -C_u w_1^m.\\
\end{split}
\label{equ::u2}
\end{equation}

Now, for $k\geq m+2$, we present the channel inputs $u(k)$ and recall the evolution of $x_s(k)$ as follows,
\begin{equation}\label{equ:u_k and x_s iterations}
\begin{split}
u(k)=& -C_u A_u^{k} (\hat{x}_{u\,0}(k-1) - x_{u\,0}) - C_s x_s(k)\\
=&-C_u A_u^{k}\bigg(  \prod_{i=m+1}^{k-1} \alpha_i {A_u^{-m-1}}{w_1^m} + \sum_{i=m+1}^{k-1} \prod_{j=i+1}^{k-1} \alpha_j \beta_i (w(i) - C_s x_s(i))\bigg) -  C_s x_s(k),\\
\text{and} \qquad &\\
x_s(k)=&A_sx_s(k-1)+B_s(u(k-1)+w(k-1)).\\
\end{split}
\end{equation}

Starting with $u(m+1) = -C_u w_1^m$ and $x_s(m+1)=0$, the above coupled iterations induce values of $u(k)$ and $x_s(k)$ that depend only on noise $w_1^{k-1}$. Therefore, for $k\geq m+2$
\begin{equation}\label{equ:input_u}
\begin{split}
u(k) \triangleq \phi_k(w_1^{k-1}),
\end{split}
\end{equation}
where the mapping $\phi_k: \mathbb{R}^{k-1} \rightarrow \mathbb{R}$ is defined by the iterations in (\ref{equ:u_k and x_s iterations}). Combining with (\ref{equ::u2}), we conclude that, for $k\geq m+1$, $u(k)$ only depends on $w_1^{k-1}$. Due to the equivalence between \textit{ Representation I} and \textit{ Representation II} of the coding scheme, this result directly holds for \textit{ Representation II}.

\subsection{Proof of Theorem \ref{thm_main}} \label{thm: main}
Following from Proposition \ref{thm_capacity_achieving_code}, the proposed coding scheme derived from the optimal filter $\mathbb{Q}$ can achieve the feedback capacity $C_{fb}$. Then, in what follows, we only need to show that under the selected initializations of the proposed coding scheme, the following secrecy requirement is satisfied,
$$\lim_{n\rightarrow \infty} \frac{1}{n}I(x_{u\,0}; z_1^n, \tilde{z}_1^n, \hat{z}_1^n)=0.$$
In this proof, we use the \textit{Representation II} of the proposed coding scheme. Following from the model in Fig. \ref{fig:problemModel} and (\ref{equ::u2}), (\ref{equ:input_u}) in the proof of Proposition \ref{prop_channel_input}, the three inputs $z(k)$, $\tilde{z}(k)$ and $\hat{z}(k)$ to Eve for $k \geq 1$ are given by
\begin{equation}\label{equ:1st_wiretap_channel}
\begin{split}
z_1^m &= u_1^m + v_1^m =A_u^{m+1} x_{u\,0} + v_1^m,\\
z(m+1) &=  u(m+1) + v(m+1) = -C_u w_1^m + v(m+1),\\
z(k) &= \phi_k(w_1^{k-1})+v(k), \quad k\geq m+2,\\
\end{split}
\end{equation}
and
\begin{equation}\label{equ:2nd_wiretap_channel}
\begin{split}
\tilde{z}_1^m &= u_1^m + \tilde{v}_1^m + w_1^n = A_u^{m+1} x_{u\,0} + \tilde{v}_1^m + w_1^n,\\
\tilde{z}(m+1) &= u(m+1) + \tilde{v}(m+1) + w(m+1) = -C_u w^m + \tilde{v}(m+1) + w(m+1),\\
\tilde{z}(k) &= \phi_k(w_1^{k-1})+ \tilde{v}(k) + w(k), \quad k\geq m+2,\\
\end{split}
\end{equation}
and
\begin{equation}\label{equ:3rd_wiretap_channel}
\begin{split}
\hat{z}_1^m &= u_1^m + \hat{v}_1^m + w_1^n = A_u^{m+1} x_{u\,0} + \hat{v}_1^m + w_1^n,\\
\hat{z}(m+1) &= u(m+1) + \hat{v}(m+1) + w(m+1) = -C_u w^m + \hat{v}(m+1) + w(m+1),\\
\hat{z}(k) &= \phi_k(w_1^{k-1})+ \hat{v}(k) + w(k), \quad k\geq m+2,\\
\end{split}
\end{equation}
where $u_1^m = A_u^{m+1}x_{u\,0}$ is the selected initializations. Recall that noises $v(k)$, $\tilde{v}(k)$ and $\hat{v}(k)$ are additive noises defined in (\ref{model: wiretap channel}).



Then, for $n \geq k+ \max\lbrace d,\tilde{d},\hat{d}\rbrace+1$ and $k\geq m+1$ we have
\begin{equation} \label{proof:entropy_ineq}
\begin{split}
h(x_{u\,0}|z_1^n, \tilde{z}_1^n,\hat{z}_1^n)\stackrel{(a)}{\geq} & h(x_{u\,0}|z_1^n, \tilde{z}_1^n,\hat{z}_1^n, w_1^{n}, v_{m+1}^n, \tilde{v}_{m+1}^n, \hat{v}_{m+1}^n)\\
\stackrel{(b)}{=} & h(x_{u\,0}|z_1^m, \tilde{z}_1^m, \hat{z}_1^m, w_1^n, v_{m+1}^n, \tilde{v}_{m+1}^n,\hat{v}_{m+1}^n)\\
\stackrel{(c)}{=} & h(x_{u\,0}|A_u^{m+1} x_{u\,0} + v_1^m, A_u^{m+1} x_{u\,0} + \tilde{v}_1^m+ w_1^m,A_u^{m+1} x_{u\,0} + \hat{v}_1^m+ w_1^m, w_1^n, v_{m+1}^n, \tilde{v}_{m+1}^n,\hat{v}_{m+1}^n)\\
= & h(x_{u\,0}|A_u^{m+1} x_{u\,0} + v_1^m, A_u^{m+1} x_{u\,0} + \tilde{v}_1^m,A_u^{m+1} x_{u\,0} + \hat{v}_1^m, w_1^n,v_{m+1}^n, \tilde{v}_{m+1}^n,\hat{v}_{m+1}^n)\\
\stackrel{(d)}{=} & h(x_{u\,0}|A_u^{m+1} x_{u\,0} + v_1^m, A_u^{m+1} x_{u\,0} + \tilde{v}_1^m, A_u^{m+1} x_{u\,0} + \hat{v}_1^m, w_1^n,v_{m+1}^{m+d}, \tilde{v}_{m+1}^{m+\tilde{d}},\hat{v}_{m+1}^{m+\tilde{d}})\\
=& h(x_{u\,0}|A_u^{m+1} x_{u\,0} + v_1^m, A_u^{m+1} x_{u\,0} + \tilde{v}_1^m, A_u^{m+1} x_{u\,0} + \hat{v}_1^m, v_{m+1}^{m+d}, \tilde{v}_{m+1}^{m+\tilde{d}},\hat{v}_{m+1}^{m+\tilde{d}}),\\
\end{split}
\end{equation}
where $(a)$ follows from the fact that conditioning does not increase entropy, step $(b)$ follows from Proposition \ref{prop_channel_input}, step $(c)$ follows from (\ref{equ:1st_wiretap_channel}), (\ref{equ:2nd_wiretap_channel}) and (\ref{equ:3rd_wiretap_channel}), step (d) follows from the finite memory of wiretap channel noise ($v$, $\tilde{v}$ ,$\hat{v}$) and the last step follows from the fact that the noise $w$ is assumed to be independent from others.

Then, we obtain
\begin{equation}
\begin{split}
&I(x_{u\,0};z_1^n, \tilde{z}_1^n, \hat{z}_1^n)\\
= & h(x_{u\,0})- h(x_{u\,0}|z_1^n, \tilde{z}_1^n, \hat{z}_1^n)\\
\leq &h(x_{u\,0})- h(x_{u\,0}|A_u^{m+1} x_{u\,0} + v_1^m, A_u^{m+1} x_{u\,0} + \tilde{v}_1^m, A_u^{m+1} x_{u\,0} + \hat{v}_1^m, v_{m+1}^{m+d}, \tilde{v}_{m+1}^{m+\tilde{d}},\hat{v}_{m+1}^{m+\tilde{d}})\\
= &I(x_{u\,0};A_u^{m+1} x_{u\,0} + v_1^m, A_u^{m+1} x_{u\,0} + \tilde{v}_1^m, A_u^{m+1} x_{u\,0} + \hat{v}_1^m, v_{m+1}^{m+d}, \tilde{v}_{m+1}^{m+\tilde{d}},\hat{v}_{m+1}^{m+\tilde{d}})\\
= & I(x_{u\,0};\mathbb{A} x_{u\,0} + \mathbb{B}), \\
\end{split}
\end{equation}
where $\mathbb{A} = [A_u^{m+1},A_u^{m+1},A_u^{m+1},\pmb{0}]^T$ ($\pmb{0}$ is an $(d+\tilde{d}+\hat{d})\times m$ zero matrix) and $ \mathbb{B} = [v_1^m,\tilde{v}_1^m,\hat{v}_1^m, v_{m+1}^{m+d}, \tilde{v}_{m+1}^{m+\tilde{d}},\hat{v}_{m+1}^{m+\hat{d}}]$. Recall that message $x_{u\,0}$ is uniformly selected from $\lbrace 1,2,\cdots, 2^{nR_s}\rbrace$ that are equally spaced in an $m$-dimensional unit hypercube. The covariance matrix of $x_{u\,0}$ is $\frac{1}{12}I_{m}$ as $n\rightarrow \infty$. Following from the fact that for a fixed covariance a vector Gaussian input distribution maximizes the mutual information, we obtain the following upper bound,
\begin{equation}\label{proof:mutual_info_ineq}
\begin{split}
\lim_{n\rightarrow \infty} \frac{1}{n}I(x_{u\,0};z_1^n, \tilde{z}_1^n, \hat{z}_1^n))&\leq \lim_{n\rightarrow \infty}\frac{1}{n}I(x_{u\,0};\mathbb{A} x_{u\,0} + \mathbb{B})\\
& = \lim_{n\rightarrow \infty}\frac{1}{n} \bigg(h(\mathbb{A} x_{u\,0} + \mathbb{B}) - h(\mathbb{A} x_{u\,0} + \mathbb{B}|x_{u\,0})\bigg)\\
& =\lim_{n\rightarrow \infty}\frac{1}{n}\bigg( h(\mathbb{A} x_{u\,0} + \mathbb{B}) - h(\mathbb{B})\bigg)\\
& \leq \lim_{n\rightarrow \infty}\frac{1}{2n}\log\det\bigg( E[\mathbb{B}\mathbb{B}^T] + \frac{1}{12}\mathbb{A}\mathbb{A}^T \bigg) - h(\mathbb{B})\\
&= 0.
\end{split}
\end{equation}
The last step follows from the fact that the right-hand side upper bound is independent from index $n$. The proof is complete.

\subsection{Proof of Corollary \ref{col:SK_robustness}} \label{col: SK}
Based on the Schalkwijk's scheme in \cite{Schalkwijk66_2}, the channel input (encoder) and the message estimate (decoder) for $k=1$ are given below by using the notations in this paper.
\begin{equation}\label{SK scheme_initial}
\begin{split}
u(1) =& A_u x_{u\,0},\\
\hat{x}_{u\,0}(1) =& A_u^{-1}y(k) = x_{u\,0}) + A_u^{-1}w(1).\\
\end{split}
\end{equation}
The dynamics of the Schalkwijk's coding scheme for $k\geq 2$ can be summarized as follows,
\begin{equation}\label{SK scheme}
\begin{split}
u(k) =& \sqrt{A_u^2-1}A_u^{k-1}(\hat{x}_{u\,0}(k-1) - x_{u\,0}),\\
\hat{x}_{u\,0}(k) = &\hat{x}_{u\,0}(k-1) - A_u^{-k-1}\sqrt{A_u^2-1}y(k),\\
\end{split}
\end{equation}
where  $A_u = \sqrt{\frac{P+\sigma_w^2}{\sigma_w^2}}$ and $\sigma_w$ is the variance of the additive white Gaussian noise in the forward channel.

In our coding scheme, for $k=1$,
\begin{equation}\label{SK scheme_initial}
\begin{split}
u(1) =& A_u^{2} x_{u\,0},\\
\hat{x}_{u\,0}(1) =& {A_u^{-2}}y(1) = x_{u\,0} + A_u^{-2}w(1),\\
\end{split}
\end{equation}
Plugging these selected parameters into (\ref{equ:input_u_orthoganal}) and (\ref{equ:hat_msg}), we have channel inputs and the message estimate of our proposed coding scheme as follows,
\begin{equation}\label{SK scheme}
\begin{split}
u(k) =& \sqrt{A_u^2-1}A_u^k(\hat{x}_{u\,0}(k-1)-x_{u\,0}),\\
\hat{x}_{u\,0}(k) = &\hat{x}_{u\,0}(k-1) - A_u^{-k-2}\sqrt{A_u^2-1}y(k).\\
\end{split}
\end{equation}
By scaling the message $x_{u\,0}$ and the corresponding estimate $\hat{x}_{u\,0}$ by factor $A_u$, we recover the dynamics of the Schalkwijk's scheme. Note that this constant scaling on the message index $x_{u\,0}$ have no effect on the reliable transmission rate and the power cost at channel input. The proof is complete.


\subsection{Proof of Theorem \ref{thm:quantized}}\label{thm: quantized}

Starting from the decoder with $A_u=2^r$ and $B_u=-1$, we have the decoding dynamics ($k\geq 2$)
\begin{equation}
\begin{split}
\hat{x}_u(k)&=A_u\hat{x}_u(k-1)+B_u (y(k-1)+q(k-1))\\
& = A_u (A_u\hat{x}_u(k-2)+B_u (y(k-2)+q(k-2)))+B_u (y(k-1)+q(k-1))\\
& = A_u^2 \hat{x}_u(k-2)+A_u B_u (y(k-2)+q(k-2))+B_u (y(k-1)+q(k-1))\\
& = \cdots\\
& = A_u^{k} \hat{x}_u(0) + B_u \sum_{i=0}^{k-1}A_u^{k-1-i}(y(i)+q(i))\\
& = B_u \sum_{i=0}^{k-1}A_u^{k-1-i}(y(i)+q(i))\\
& = -\sum_{i=0}^{k-1}2^{r(k-1-i)}(y(i)+q(i)).\\
\end{split}
\end{equation}
Then, the estimate of the initial state of the encoder (i.e., the message index) is
\begin{equation}\label{equ: Marin_decoder}
\begin{split}
\hat{x}_{u\,0}(k-1)=  A_u^{-k}\hat{x}_u(k)  =& -\sum_{i=0}^{k-1}2^{-r(i+1)}(y(i)+q(i)).\\
\end{split}
\end{equation}

Next, based on (\ref{equ:input_u_orthoganal}) with $C_s =0$ and $C_u=A_u-\frac{1}{A_u}$, we have the dynamics of channel inputs as
\begin{equation}\label{equ:Marin_encoder}
\begin{split}
u(k)=& C_u A_u^{k} (x_{u\,0}- \hat{x}_{u\,0}(k-1))\\
=& C_u A_u^{k} (x_{u\,0}- A_u^{-k}\hat{x}_u(k))\\
=& C_u A_u^{k} (x_{u\,0}- A_u^{-k}(A_u\hat{x}_u(k-1)+B_u (y(k-1)+q(k-1))))\\
=& C_u A_u^{k} (x_{u\,0}- A_u^{-k+1}\hat{x}_u(k-1)) +C_u B_u (y(k-1)+q(k-1))\\
\stackrel{(a)}{=}& A_u u(k-1) +C_u B_u (y(k-1)+q(k-1))\\
=& 2^r u(k-1)+ (2^{-r}-2^r)(y(k-1)+q(k-1)),\\
\end{split}
\end{equation}
where step (a) follows from the above first two steps. Notice that our coding schemes of the decoder (\ref{equ: Marin_decoder}) and the encoder (\ref{equ:Marin_encoder}) are identical to the coding schemes (2) and (4) in \cite{Martins08}. In addition, Theorem 3.2 in \cite{Martins08} shows that, for a given $r_v$ (Definition \ref{def_rv}), the proposed scheme can achieve any transmission rate $r$ with $r<r_v$.

Now, we need to prove that the proposed coding scheme achieves secrecy with regard to the eavesdropper. In fact, we can directly follow the proof of Proposition \ref{prop_channel_input} and characterize the channel inputs as
\begin{equation}
\begin{split}
u(1) =& A_u^2 x_{u\,0}, \qquad  u(2) = -C_u (w(1)+q(1)),\\
u(k)=&-C_u A_u^{k}\bigg( (1 - A_u^{-1}B_u C_u)^{k-2}\frac{w(1)}{A_u^2} + \sum_{i=2}^{k-1} A_u^{-i-1}(1 - A_u^{-1}B_u C_u)^{k-1-i} B_u (w(i)+q(i))\bigg), \quad k\geq 3.\\
\end{split}
\end{equation}
As a consequence, we note that the channel inputs of the proposed coding scheme only depend on the past forward channel noise $w$ and feedback quantization noise $q$. This fact enables us to show that, by following the proof of Theorem \ref{thm_main}, this coding scheme satisfies the secrecy requirement
$$\lim_{n\rightarrow \infty} \frac{1}{n}I(x_{u\,0}; z_1^n, \tilde{z}_1^n,\hat{z}_1^n)=0.$$

To avoid redundancy, we herein only provide sketchy arguments. Details can be directly obtained by following (\ref{proof:entropy_ineq}) to (\ref{proof:mutual_info_ineq}).
\begin{equation}
\begin{split}
h(x_{u\,0}|z_1^n, \tilde{z}_1^n, \hat{z}_1^n)\geq & h(x_{u\,0}|z_1^n, \tilde{z}_1^n, \hat{z}_1^n, w_1^n,q_1^n, v_{2}^n, \tilde{v}_{2}^n,\hat{v}_{2}^n)\\
= & h(x_{u\,0}|z_1, \tilde{z}_1,\hat{z}_1, w_1^n, q_1^n, v_{2}^n, \tilde{v}_{2}^n, \hat{v}_{2}^n)\\
= & h(x_{u\,0}|A_u^{2} x_{u\,0} + v_1, A_u^{2} x_{u\,0} + \tilde{v}_1+ w_1, A_u^{2} x_{u\,0} + \hat{v}_1+ w_1, w_1^n,q_1^n, v_{2}^n, \tilde{v}_{2}^n, \hat{v}_{2}^n)\\
= & h(x_{u\,0}|A_u^{2} x_{u\,0} + v_1, A_u^{2} x_{u\,0} + \tilde{v}_1, A_u^{2} x_{u\,0} + \hat{v}_1, w_1^n,q_1^n, v_{2}^n, \tilde{v}_{2}^n, \hat{v}_{2}^n)\\
= & h(x_{u\,0}|A_u^{2} x_{u\,0} + v_1, A_u^{2} x_{u\,0} + \tilde{v}_1, A_u^{2} x_{u\,0} + \hat{v}_1, w_1^n,q_1^n, v_{2}^{1+d}, \tilde{v}_{2}^{1+\tilde{d}}, \hat{v}_{2}^{1+\hat{d}})\\
= & h(x_{u\,0}|A_u^{2} x_{u\,0} + v_1, A_u^{2} x_{u\,0} + \tilde{v}_1, A_u^{2} x_{u\,0} + \hat{v}_1, v_{2}^{1+d}, \tilde{v}_{2}^{1+\tilde{d}}, \hat{v}_{2}^{1+\hat{d}}).\\
\end{split}
\end{equation}

Then, we obtain
\begin{equation}
\begin{split}
&I(x_{u\,0};z_1^n, \tilde{z}_1^n, \hat{z}_1^n)\\
\leq &I(x_{u\,0};A_u^{2} x_{u\,0} + v_1, A_u^{2} x_{u\,0} + \tilde{v}_1, A_u^{2} x_{u\,0} + \hat{v}_1, v_{2}^{1+d}, \tilde{v}_{2}^{1+\tilde{d}}, \hat{v}_{2}^{1+\hat{d}})\\
= & I(x_{u\,0};\mathbb{A} x_{u\,0} + \mathbb{B}), \\
\end{split}
\end{equation}
where $\mathbb{A} = [A_u^2,A_u^{2},A_u^{2},\pmb{0}]^T$ ($\pmb{0}$ is an $(d+\tilde{d}+\hat{d})\times 1$ zero matrix) and $ \mathbb{B} = [v_1,\tilde{v}_1,\hat{v}_1,v_{2}^{1+d}, \tilde{v}_{2}^{1+\tilde{d}},\hat{v}_{2}^{1+\hat{d}}]$. The rest of the proof is omitted as it directly follows from (\ref{proof:mutual_info_ineq}).


\section{Conclusion} {\label{sec:conclusion}}
In this paper, we considered the ARMA(k) Gaussian wiretap channel with feedback and showed that a variant of the generalized $S$-$K$ scheme, which is a feedback capacity-achieving code, secures transmissions by itself from the eavesdropper. Namely, the feedback secrecy capacity equals the feedback capacity without the presence of an eavesdropper. We further extended our scheme to the AWGN channel with quantized feedback and proved that our scheme can achieve a positive secrecy rate, which converges to the AWGN channel capacity as the quantization noise decreases to zero.

We conclude this paper by listing a few related research topics, which can facilitate to illustrate the complete picture of the secrecy communications with feedback. First of all, it is known that the $S$-$K$ coding scheme nicely unifies communications, control and estimation for feedback systems. In this paper, by leveraging the tools from both control and communications, we showd that (a variant of) the $S$-$K$ scheme automatically provides secrecy for the legitimate users. Therefore, understanding the secrecy nature of the $S$-$K$ scheme from an estimation perspective could be a missing piece of the work. One possible investigation along this line is to extend the fundamental relation between the derivative of the mutual information and the MMSE \cite{Guo_IMMSE}, known as I-MMSE, from the open-loop channels to the feedback channels by invoking the direct information \cite{Massey1990, Kramer_thesis} rather than the mutual information \cite{Weissman_feedback}. Furthermore, extend the current results to channels with noisy feedback can be very valuable. Toward this end, it is necessary to construct a feedback coding scheme with a good achievable rate for noisy feedback channels with no eavesdropper, which itself is a quite nontrivial problem in general. Finally, extensions to the multi-access colored Gaussian channels with feedback can be of much interest to the community in the field of secrecy communications.

\bibliographystyle{IEEEtran}
\bibliography{ref}

\end{document}